\documentclass[11pt]{article}
\usepackage[margin=1in]{geometry}
\usepackage{amsmath,amssymb,amsthm}
\usepackage{booktabs}
\usepackage[round]{natbib}
\setcitestyle{aysep={}}
\usepackage[hidelinks]{hyperref}

\newtheorem{proposition}{Proposition}
\newtheorem{lemma}{Lemma}
\theoremstyle{definition}
\newtheorem{definition}{Definition}

\title{Distinguishable Category Groupings in Complex Surveys: Which Categories Can the Design Tell Apart?}
\author{Jiarui Qi \and Yinhai Wang}
\date{}

\begin{document}
\maketitle

\begin{abstract}
Surveys often estimate the share of cases in each of several categories, such as the causes of an event or the reasons for a visit. Readers then compare the categories, but reliability rules usually check each estimate alone, not the differences between them. We ask instead which groups of categories can be told apart under the sampling design. The categories often fall into blocks, such as driver-related and vehicle-related causes. A grouping is called distinguishable when every pair of its groups in the same block differs by more than the design's minimum detectable difference, computed from the primary sampling
units at the design's degrees of freedom. The maximal distinguishable groupings are the most detailed and can differ in size, while merging two groups can destroy distinguishability. So even if splitting any one group in two breaks distinguishability, a more detailed distinguishable grouping may still exist. We propose two methods. One merges categories until the grouping is distinguishable, then tests every single split. The other lists every distinguishable grouping of a block with few categories. When groups in different blocks must also differ, there might be no grouping that qualifies. An application to the National Motor Vehicle Crash
Causation Survey, with 24 primary sampling units in 12 strata, gives the grouping maxima for two blocks of its categories. Of more than 27 million groupings of its 13 vehicle-related critical
reasons, none with more than four groups is distinguishable, and its environment-related reasons resolve into
at most four groups. Merging two groups of a distinguishable vehicle grouping destroys distinguishability in
5{,}239 of 31{,}773 cases. If groups from different blocks must also be told apart, no grouping qualifies.
Analysts can apply these methods before publishing categorical comparisons from designs with few degrees of
freedom.
\end{abstract}

\paragraph{Statement of Significance.}
Official statistics and research surveys publish estimates of the share of cases in each category, and readers treat differences between those categories as findings. With few sampling units, a survey may not be able to tell many of them apart, even if each estimate looks reliable. This article defines the set of most detailed groupings
of categories that a survey's design can distinguish, and shows why one number of groups cannot summarize it. It
also gives methods to find that set, one of them exact for short lists of categories. Another, a quick scan, applies
when groups from different parts of the list must also differ, and it can prove that no grouping works. Applied to a national crash investigation survey,
the methods show that thirteen vehicle-related crash reasons can be told apart in at most four groups. When groups
from different parts of the list must also differ, no grouping of all the survey's crash reasons works. The
same check can be run on any survey table before categorical comparisons are published.


\section{Introduction}

Survey reports often sort cases into the categories of a taxonomy, for example, crashes by cause or medical visits
by reason. Readers rank whole groups of reasons based on which is more common. Whether the survey supports such a comparison depends on its design. A design with few
primary sampling units (PSUs) supports fewer distinctions than its sample size suggests, because design-based variances are computed from the PSU totals and carry about as many degrees of freedom as there are PSUs minus strata
\citep{KornGraubard1999}.

Existing rules answer narrower questions. Data presentation standards check one estimate at a time to decide whether it is reliable enough to publish \citep{Parker2017,Parker2023}. They say nothing about which categories differ.
Tree and merging procedures form a single grouping by successive tests \citep{Kass1980,Opsomer2025}, and
sequential aggregation merges categories by information loss instead \citep{Jackson2008}. Procedures for
comparisons and rankings test which estimates differ \citep{Griffiths1992,Kott2023} or sort them into report card grades \citep{Kline2024}. A joint confidence region for a whole ranking can be built from the estimates' confidence
intervals, with no knowledge of the sampling design \citep{Klein2020}. Simultaneous intervals can also split a
ranking into its finest ordered groups, although no group is re-estimated as one merged category \citep{Rising2023}. Partition searches fix the number of groups and optimize an association criterion
\citep{Carrizosa2023}, or take the largest number of groups whose worst pair still differs \citep{Eo2013}. Work on official statistics also weighs geographic detail against out-of-sample error \citep{Dunn2026}. None of this work gives the set of partitions of a taxonomy that the design itself can distinguish.

In this article, we define that set, the maximal distinguishable groupings, for a stratified multistage design and a taxonomy whose categories may be merged only within fixed blocks, such as the driver-related or the vehicle-related reasons. The criterion is design-based and uses the design's degrees of freedom. The construct is a set of incomparable
partitions rather than a single number of groups. We also show that merging two groups of a distinguishable grouping can leave a grouping that is not distinguishable. So even if splitting any one group in two breaks distinguishability, a more detailed distinguishable grouping may still exist. For blocks with few categories, listing every distinguishable grouping gives the maximal ones exactly. When all pairs of groups must differ, a subset scan can prove
that no partition qualifies. To our knowledge, no existing method defines or
computes the most detailed partitions that a survey design can distinguish under a design-based criterion. Section 4
applies the construction to a national crash investigation survey, and the Supplementary Material gives all
proofs and algorithms.

\section{The Construct}

\subsection{Contrasts at the design's degrees of freedom}

Consider a sample of PSUs drawn in $H$ strata, $n_h$ in stratum $h$, with $n=\sum_h n_h$ and design weights
$w_i$. Every unit $i$ of the frame falls in at most one category of a taxonomy $C=\{1,\dots,m\}$. Let $\hat X$
be the weighted total of the frame and $\hat Y_k$ the weighted total of category $k$, so that
$\hat p_k=\hat Y_k/\hat X$. A cell, or group, is a set of categories $A\subseteq C$, with share
$\hat p_A=\sum_{k\in A}\hat p_k$. For disjoint cells $A$ and $B$, the contrast is
$\hat\Delta(A,B)=\hat p_A-\hat p_B$. Its Taylor linearization variance $\hat v(A,B)$ is computed from the PSU
totals of the linearized variable \citep{Binder1983}, as set out in Supplement S1. Because the contrast has the weighted total of the whole frame as its denominator, every PSU enters its variance, whether or not it holds units of either cell. As a result, each contrast
carries $\nu=n-H$ degrees of freedom, with no domain adjustment. A merged cell's variance is recomputed
from the PSU totals, never assembled from its members' variances alone, since it also carries their
covariances.

\subsection{The criterion}

Two cells are \emph{distinguishable} when the observed gap exceeds the minimum detectable difference (MDD),
\begin{equation}
|\hat\Delta(A,B)| \;>\; \bigl(t_{\nu,1-\alpha/2}+t_{\nu,1-\beta}\bigr)\,\widehat{\mathrm{SE}}(A,B),
\label{eq:mdd}
\end{equation}
that is, when the separation ratio $r(A,B)=|\hat\Delta|/\widehat{\mathrm{SE}}$ exceeds
$c=t_{\nu,1-\alpha/2}+t_{\nu,1-\beta}$. A pair with $\widehat{\mathrm{SE}}=0$ is not distinguishable. With a
two-sided $\alpha=0.05$, power $1-\beta=0.80$ and $\nu=12$, $c=2.1788+0.8726=3.0514$. A Wald test
($c=t_{\nu,1-\alpha/2}$) and its Benjamini--Hochberg adjustment \citep{BenjaminiHochberg1995} give two alternative criteria. The MDD criterion implies the Wald criterion pair by pair, so every partition it accepts is accepted by the Wald
test as well, and the largest number of cells is never higher under the MDD criterion (Supplement S2).

Algebraically, (\ref{eq:mdd}) is a Wald test at a stricter level, about 0.01 two-sided when $\nu=12$ and the power is 0.80. However, a pair below its minimum detectable difference is one this design could not be expected to tell apart at the stated power \citep[compare][]{SinghWestlake2005}. That is a statement about the design, not evidence that the two shares are equal. Both quantiles in $c$ come from the $t$ distribution on $\nu$ degrees of freedom rather than from the normal distribution, since the standard error is itself estimated.

\subsection{Maximal distinguishable groupings}

The categories usually fall into blocks $G_1,\dots,G_L$ set by the taxonomy, and a merged cell must stay inside one block. Some categories, such as a code for an undetermined cause, may not be merged at all.

\begin{definition}
A partition $\pi$ of $C$ is \emph{admissible} when each cell lies in one block and every unmergeable category is a
singleton. An admissible $\pi$ is \emph{distinguishable} under scope W (within-block pairs) when every pair of its cells in the same block is distinguishable, and under scope A (all pairs) when every pair of its cells is. A partition $\pi$ \emph{refines} $\sigma$, and $\sigma$ is a \emph{coarsening} of $\pi$, when each cell of $\pi$ lies inside a cell of $\sigma$. A distinguishable $\pi$ is \emph{maximal} when no strict
refinement of $\pi$ is distinguishable. A partition is also called a \emph{grouping}, and $K(\pi)$ is its number of cells.
\end{definition}

Two distinct maximal partitions never refine one another. We write $K_{\max}$ for the largest number of cells in a distinguishable partition. For one block, $K_{\max}$ is taken over the partitions of that block. It describes only the maximal partitions with that many cells.

\section{Properties and Algorithms}

\begin{lemma}\label{lem:local}
Whether two cells are distinguishable depends on those two cells alone.
\end{lemma}

The contrast and its variance are functions of the two cells' PSU totals and of the frame's PSU totals, and none
of these changes with the rest of the partition. Under scope W the lemma splits the problem by block. The cells of a partition that lie in one block form its factor in that block. A partition
is distinguishable if and only if each block's factor is, and maximal if and only if every factor is maximal in
its block (Supplement S3). The maximal partitions are therefore exactly those formed by choosing one maximal factor in each block, and the numbers of cells add across blocks.

\subsection{Coarsening can destroy distinguishability}

\begin{proposition}\label{prop:nonmono}
A distinguishable partition can have an admissible coarsening that is
not distinguishable.
\end{proposition}

The share of a merged cell is the sum of its members' shares, and its variance includes their covariances.
Merging can move the merged share toward the share it is compared with. It can also raise the standard error,
for instance when the members' estimates are positively correlated across PSUs. Either change lowers the
separation ratio. Section 4 shows an instance in real data, and Supplement S3 constructs another. Shares should therefore be reported at a partition that is itself distinguishable, never at an unchecked coarsening of one.

\begin{proposition}\label{prop:cert}
Call a distinguishable $\pi$ \emph{certified} with cap $\bar m$ when no cell has more than $\bar m$ categories and no partition
obtained by splitting one cell of $\pi$ into two is distinguishable.
(a) A maximal partition whose cells have at most $\bar m$ categories is certified.
(b) If coarsening never destroys distinguishability among admissible partitions, every certified partition is maximal.
(c) Under scope W, a certified block factor with at least $K_{\max}-1$ cells is maximal whenever that block's
$K_{\max}$ is known exactly.
(d) Otherwise a certified partition need not be maximal.
\end{proposition}

Part (d) is possible because coarsening can destroy distinguishability (Proposition \ref{prop:nonmono}). As the four-category example of Supplement S3 shows, a distinguishable refinement can lie two splits away
while every single split fails. The certificate is cheap to run, and under the conditions of part (c) it proves that a block factor is maximal. Since maximal partitions often fall below $K_{\max}$ in practice (Section 4), the number of cells in one partition does not show $K_{\max}$.

\subsection{Finding maximal distinguishable groupings}

A greedy search starts from the all-singletons partition. While some mergeable pair in scope fails the criterion, it merges the failing pair with the smallest separation ratio and recomputes every affected contrast
from the PSU totals. It stops when no such pair fails. The search belongs to an established family of procedures that merge by successive tests
\citep{Kass1980,Fan2009}. A survey version of Kass's procedure runs its split tests under the design
\citep{Opsomer2025}.

The terminal partition of the search is then checked with the certificate of Proposition~\ref{prop:cert}. For every cell $A$ of at most $\bar m$ categories, each of its $2^{|A|-1}-1$ splits in two gives a refined partition, and every in-scope pair of cells in that partition is tested against the criterion. A surviving split is applied and the check restarts. A terminal partition with a cell above the cap is reported as uncertified and is not counted as maximal. The third step of
\citet{Kass1980} also re-splits a merged category, but it tests only the two parts against each other and
does not claim that its result is maximal.

For blocks with few categories, every distinguishable partition can be listed. For a block of $b$ categories, all $2^b$ cell totals can be
computed once. A walk visits each admissible partition of the block exactly once. At each step it takes the lowest-numbered category not yet in a cell and tries every admissible subset of the remaining categories that contains this category. The choices made up to any step form a branch of the walk.

\begin{proposition}\label{prop:enum}
Under scope W, abandoning a branch of this walk at its first indistinguishable pair loses no distinguishable
partition. The walk with this rule returns exactly the block's distinguishable partitions.
\end{proposition}

By Lemma \ref{lem:local} the failing pair fails in every partition that keeps the cells chosen so far, and scope W requires every within-block pair to be distinguishable. The walk thus yields the exact $K_{\max}$ and every maximal partition of the block, without visiting all of its partitions, whose count is the Bell number of $b$ \citep{Rota1964}. For the environment block of Section 4, which has 10{,}480{,}142{,}147 partitions, the walk needed 81{,}303{,}674 pair tests.

Under scope A, an unmergeable category can make the search unnecessary.

\begin{proposition}\label{prop:orphan}
Let $F$ be an unmergeable category. Under scope A, if some other category $k$ lies in no admissible cell that is
distinguishable from $\{F\}$, then no admissible partition is distinguishable.
\end{proposition}

Call such a $k$ an orphan. Every admissible partition contains the cell $\{F\}$ and a cell holding $k$, and that
cell is a subset of $k$'s block. Scanning the subsets of that block which contain $k$ therefore rules out every admissible partition at once.

\section{Illustration: A Crash Investigation Survey}

The National Motor Vehicle Crash Causation Survey (NMVCCS) of the National Highway Traffic Safety
Administration (NHTSA) sent researchers to crash scenes, and its nationally representative sample covers crashes
from July 3, 2005, to December 31, 2007 \citep{NHTSA2008a}. A crash qualified when emergency medical services
were dispatched and a light passenger vehicle was towed because of damage, among other conditions. Researchers
were on duty between 6 a.m. and midnight, so the frame covers crashes in those hours only. The NMVCCS analytical user's manual counts 6{,}949 investigated crashes, of which 5{,}470
have positive case weights that sum to 2{,}188{,}970 crashes \citep{NHTSA2008b}. The design has 24 PSUs in 12 strata
\citep{Choi2008}, so $\nu=12$.

The critical reason is the immediate reason for the event that made the crash imminent \citep{NHTSA2008a}. It is coded on one vehicle of a crash, and occasionally on two. Our crash-level category is
that vehicle's code, taken from the lower-numbered vehicle when two vehicles carry one. The 59 categories
observed at crash level fall in four blocks: 29 driver-related, 13 vehicle-related
and 16 environment-related reasons, plus one code for an undetermined reason that is never merged. We apply
criterion (\ref{eq:mdd}) with $\alpha=0.05$ and power 0.80 ($c=3.0514$), under scope W unless stated otherwise.

\begin{table}[t]
\centering
\caption{Results of listing every distinguishable partition of the vehicle and environment blocks (MDD criterion, power 0.80, scope W).}
\label{tab:census}
\begin{tabular}{lrr}
\toprule
 & Vehicle & Environment \\
\midrule
Categories & 13 & 16 \\
Partitions of the block & 27{,}644{,}437 & 10{,}480{,}142{,}147 \\
Distinguishable partitions & 11{,}015 & 97{,}458 \\
\quad with 1 / 2 cells & 1 / 1{,}704 & 1 / 15{,}970 \\
\quad with 3 / 4 cells & 8{,}597 / 713 & 81{,}415 / 72 \\
\quad with 5 or more cells & 0 & 0 \\
$K_{\max}$ & 4 & 4 \\
One-step coarsenings of distinguishable partitions & 31{,}773 & 260{,}647 \\
\quad of which not distinguishable & 5{,}239 & 376 \\
\bottomrule
\end{tabular}
\end{table}

Table \ref{tab:census} counts every distinguishable partition of the vehicle and environment blocks. So no grouping of the
13 vehicle-related reasons into more than 4 groups is distinguishable under this design,
while 713 groupings into 4 are. The exact maximum is also 4 for the 16
environment-related reasons.

Coarsening often destroys distinguishability among the vehicle-related reasons. In the vehicle block the partition
\[
\{200, 202, 204, 205, 210, 211, 212, 213, 299\}\;|\;\{201, 206, 208\}\;|\;\{203\}
\]
is distinguishable, with separation ratios 4.139, 5.615 and 4.288. Merging its last two cells leaves one pair, which separates at 2.519, below
3.0514. Merging two cells of a partition gives a one-step coarsening. Of the 31{,}773 one-step coarsenings of distinguishable vehicle partitions, 5{,}239 are not distinguishable. They come from 4{,}262 of the 11{,}014 distinguishable partitions with at least two cells.

Maximal partitions below $K_{\max}$ are the rule. A distinguishable strict refinement of a three-cell vehicle partition must have four cells, and a four-cell partition coarsens into only six three-cell partitions. So at most six times 713 of the 8{,}597 distinguishable three-cell vehicle partitions have such a refinement. All the others are maximal at three cells, although the block's maximum is four
(Proposition~\ref{sprop:antichain} in the Supplement). In the environment block, at most six times 72 of the 81{,}415 distinguishable three-cell partitions are not maximal.

\begin{table}[t]
\centering\small
\caption{Orphan scan against the undetermined code, by criterion (scope A). The vehicle block has 8{,}191
nonempty subsets and the environment block 65{,}535. Passing: subsets distinguishable from the undetermined
code. Covered: categories lying in at least one passing subset. D1 is the MDD criterion (\ref{eq:mdd}), D2 the
Wald test and D3 its Benjamini--Hochberg adjustment at false discovery rate $q=0.05$ (Supplement S2).}
\label{tab:orphan}
\begin{tabular}{llrrp{5.2cm}}
\toprule
Criterion & Block & Passing & Covered & Orphans \\
\midrule
D1, power 0.80 & vehicle & 109 & 7 of 13 & 200, 202, 203, 206, 208, 210 \\
 & environment & 2{,}155 & 13 of 16 & 509, 522, 525 \\
D1, power 0.90 & vehicle & 0 & 0 of 13 & all 13 \\
 & environment & 0 & 0 of 16 & all 16 \\
D2 (Wald) & vehicle & 1{,}212 & 12 of 13 & 202 \\
 & environment & 22{,}518 & 15 of 16 & 509 \\
D3 (Benjamini--Hochberg) & vehicle & 689 & 12 of 13 & 202 \\
 & environment & 16{,}781 & 15 of 16 & 509 \\
\bottomrule
\end{tabular}
\end{table}

Scope A also requires every cell to be distinguishable from the undetermined code. Of the 8{,}191 nonempty
subsets of the vehicle block, 109 meet this requirement, and together they cover only
7 of its 13 categories. Categories 200, 202, 203, 206, 208, 210 lie in no such subset, so they are orphans. By
Proposition \ref{prop:orphan}, no admissible partition of the 59 categories is distinguishable under scope A.
The same holds at power 0.90 and under the Wald and Benjamini--Hochberg criteria (Table~\ref{tab:orphan}). The driver block has $7.13\times 10^{22}$ partitions, far too many for the listing method, but it never has to be examined. The
illustration thus settles all 59 categories under scope A, and it gives exact maxima for the vehicle and environment blocks (Table~\ref{tab:census}).

\section{Discussion}

Maximal distinguishable groupings turn a vague caution about small samples into a statement that can be checked. A table that
compares categories should do so at a partition the design can distinguish, and it should name that partition.
Before any search, a scan against an unmergeable category can show that no partition is distinguishable under scope A. Where a block has few enough categories, the exact maximum can be reported, and where it does not, the certificate can show that splitting any one group in two breaks distinguishability. None of these steps needs more than the PSU totals that design-based variance estimation
already uses, so the methods apply to any stratified multistage design.

The degrees of freedom of every contrast equal the number of variance PSUs minus the number of strata. By this count the Crash Report Sampling System, a current NHTSA survey, has 42 degrees of freedom on its full
design. We counted the variance PSUs and strata in its public files \citep{CRSSdata}, from the design variables that its user's manual describes \citep{NHTSA2024crss}.

Every input these methods need can be fixed in advance, so a reported partition cannot be tuned to the estimates. In the illustration, the criterion and its parameters, the block constraint, the
crash-level coding rule, the search and the certificate were all fixed before any partition was computed, as was the possibility that no partition is distinguishable under scope A. The complete listing for the vehicle block and the orphan scan were added after the planned analysis had run, and the listing for the environment block was declared before being run.

Models are a natural alternative. Small-area and hierarchical estimators borrow strength across categories and
can report more detailed groupings. The extra detail, however, comes from the model rather than from the design. The maximal distinguishable groupings
state which categorical comparisons the design supports on its own, and they give a benchmark for model-based claims about categories.

Our construction has limitations. First, the standard error in (\ref{eq:mdd}) is estimated, so the maximal distinguishable groupings share that uncertainty.
Second, the criterion is pairwise. Adjusting $\alpha$ for multiplicity, or using fewer effective degrees of freedom,
would only enlarge $c$, and a larger $c$ can never raise $K_{\max}$ (Supplement S2). In the illustration, $c$
already exceeds the cut that the Benjamini--Hochberg procedure implies for the Wald test over all category pairs.

\section*{Data Availability}
The NMVCCS and CRSS public-use files are available from NHTSA \citep{NMVCCSdata,CRSSdata}. The code for the analyses in Section 4
will be released with a later version of this article. The analysis of Section 4 was pre-registered, and Section 5 states which of its elements were fixed in advance.

\section*{Acknowledgments}
The authors used OpenAI ChatGPT to help edit the manuscript to improve the grammar and readability, assist verifying mathematical correctness, and check and debug analysis and .tex code. The authors reviewed and verified all output and take full responsibility for the content of this article.

\clearpage
\begin{center}
{\large\bfseries Supplementary Material}
\end{center}

\setcounter{section}{0}
\renewcommand{\thesection}{S\arabic{section}}
\setcounter{proposition}{0}
\renewcommand{\theproposition}{S\arabic{proposition}}
\setcounter{lemma}{0}
\renewcommand{\thelemma}{S\arabic{lemma}}
\setcounter{table}{0}
\renewcommand{\thetable}{S\arabic{table}}
\setcounter{equation}{0}
\renewcommand{\theequation}{S\arabic{equation}}
\renewcommand{\theHsection}{S\arabic{section}}
\renewcommand{\theHproposition}{S\arabic{proposition}}
\renewcommand{\theHlemma}{S\arabic{lemma}}
\renewcommand{\theHtable}{S\arabic{table}}
\renewcommand{\theHequation}{S\arabic{equation}}

\section{Variance and Degrees of Freedom}

Let $(h,j)$ index PSU $j$ of stratum $h$. For disjoint cells $A$ and $B$ and a unit $i$ of the frame, put
$u_i=\mathbf 1\{i\in A\}-\mathbf 1\{i\in B\}-\hat\Delta(A,B)$. The linearized variable of the ratio contrast is
$z_i=w_iu_i/\hat X$ \citep{Binder1983}, and its PSU totals are $z_{hj}=\sum_{i\in(h,j)}z_i$. Under the usual with-replacement
approximation for the first stage,
\begin{equation}
\hat v(A,B)=\sum_{h=1}^{H}\frac{n_h}{n_h-1}\sum_{j=1}^{n_h}\bigl(z_{hj}-\bar z_h\bigr)^2,\qquad
\bar z_h=\frac1{n_h}\sum_{j}z_{hj},
\label{eq:var}
\end{equation}
with $\nu=n-H$ degrees of freedom \citep{KornGraubard1999}. A unit in neither cell has
$u_i=-\hat\Delta(A,B)$, so every PSU with frame units enters (\ref{eq:var}), whether or not it holds units of $A$ or
$B$. The contrast is a ratio of totals over the whole frame, not a domain estimate, and no domain reduction of
$\nu$ applies. Write $\mathbf 1_A\in\{0,1\}^m$ for the indicator of the categories of $A$, and let $z_i^{(k)}=w_i(\mathbf 1\{i\in k\}-\hat p_k)/\hat X$ be the linearized variable of $\hat p_k$. Then $z_i=\sum_k(\mathbf 1_A-\mathbf 1_B)_k\,z_i^{(k)}$, so
$\hat v(A,B)=(\mathbf 1_A-\mathbf 1_B)^\top\hat{\mathbf V}(\mathbf 1_A-\mathbf 1_B)$, where $\hat{\mathbf V}$ is the estimated covariance matrix of the share estimators, computed by applying (\ref{eq:var}) to the vectors $(z_i^{(1)},\dots,z_i^{(m)})$. A merged cell's variance therefore includes its members' covariances and must be recomputed from the PSU totals rather than from the members' variances.

\section{The Three Criteria}

Write $r(A,B)=|\hat\Delta(A,B)|/\widehat{\mathrm{SE}}(A,B)$ and treat a pair with $\widehat{\mathrm{SE}}=0$ as not distinguishable.
\begin{itemize}
\item[D1] (minimum detectable difference) $r>c_1=t_{\nu,1-\alpha/2}+t_{\nu,1-\beta}$.
\item[D2] (Wald) $r>c_2=t_{\nu,1-\alpha/2}$.
\item[D3] (Benjamini--Hochberg) at the category level, the Benjamini--Hochberg procedure is applied to D2's two-sided $t_\nu$ tests over the family of all category pairs, and a pair passes when its adjusted $p$-value is at most $q\le\alpha$ ($q=0.05$ in the illustration). At the cell level, D3 uses the fixed cut $c_3$ on $r$ that the procedure implies for that family.
\end{itemize}

\begin{proposition}\label{sprop:nest}
For $\beta<1/2$, $c_1>c_2$, so D1 implies D2 for every pair of cells. D1 at power 0.90 implies D1 at power 0.80.
Hence every partition distinguishable under the stricter criterion is distinguishable under the looser one, and
$K_{\max}$ can only fall as the criterion tightens.
\end{proposition}

\begin{proof}
The quantile $t_{\nu,1-\beta}$ is positive when $\beta<1/2$ and increases with $1-\beta$. A partition is distinguishable exactly when every in-scope pair of its cells is distinguishable, and the stricter criterion implies the looser one pair by pair.
\end{proof}

\begin{proposition}\label{sprop:bh}
At the category level D3 implies D2, since an adjusted $p$-value is never below the raw one. At the cell level D1
implies D3 exactly when $c_3\le c_1$. This holds in the illustration, where $c_3=2.436<3.0514=c_1$, but not in general, because with few rejections $c_3$ approaches a Bonferroni cut, which can exceed $c_1$.
\end{proposition}

\section{Structural Results and Proofs}

\begin{lemma}[Lemma~\ref{lem:local} of the main text]\label{slem:local}
Whether cells $A$ and $B$ are distinguishable depends on $A$ and $B$ alone.
\end{lemma}

\begin{proof}
By (\ref{eq:var}), $\hat\Delta(A,B)$ and $\hat v(A,B)$ are functions of the PSU totals of the indicators of $A$ and
$B$ and of the frame's PSU totals, which sum to $\hat X$. None of these depends on how the remaining categories are grouped, and $\nu$
is fixed by the design.
\end{proof}

Lemma \ref{slem:local} would not hold if the contrasts were normalized by a quantity that depends on the partition, for instance a denominator restricted to the categories in play.

\begin{proposition}[blocks]\label{sprop:blocks}
Under scope W, suppose merges are confined to the same blocks that define the scope. (i) A partition is
distinguishable if and only if each of its block factors is distinguishable in its block. (ii) Its number of cells is the sum of its factors' numbers of cells. (iii) It is maximal if and only if every factor is maximal in its block.
\end{proposition}

\begin{proof}
(i) The pairs that scope W tests are the within-block pairs, which are the disjoint union of the factors' pairs. By Lemma \ref{slem:local}, whether each pair is distinguishable depends on that pair alone. (ii) is immediate. (iii) If every factor is maximal and
$\sigma$ is a distinguishable strict refinement, some factor of $\sigma$ strictly refines the corresponding factor
and is distinguishable by (i), a contradiction. If some factor is not maximal, replacing it by a distinguishable
strict refinement gives, by (i), a distinguishable strict refinement of the whole partition.
\end{proof}

Merging is allowed only within sets of categories called merge sets, each inside one block. The decomposition of Proposition~\ref{sprop:blocks} does not extend to merge sets smaller than the blocks of the scope. Cells of different merge sets within one block are compared but can never be merged, so choosing one distinguishable partition of each merge set need not give a distinguishable partition of the block.

\begin{proposition}[the search]\label{sprop:search}
The greedy search of Section 3.2 stops after at most $m-1$ merges and is deterministic once ties are broken in a
fixed order. At its end every in-scope pair of cells that could be merged is distinguishable. When every in-scope pair of cells is mergeable, as when the merge sets coincide with the scope's blocks and
each unmergeable category is a block of its own, the terminal partition is distinguishable. Otherwise pairs of cells from different merge sets may still fail, and such failures are a property of that terminal partition only.
\end{proposition}

\begin{proof}
Each merge lowers the number of cells by one. The stopping rule is that no mergeable in-scope pair fails. When
every in-scope pair is mergeable, this is the definition of a distinguishable partition.
\end{proof}

\begin{proposition}[Proposition~\ref{prop:cert} of the main text]\label{sprop:cert}
(a) A maximal partition whose cells have at most $\bar m$ categories is certified. (b) If coarsening never
destroys distinguishability among admissible partitions, every certified partition is maximal. (c) Under scope W,
a certified block factor with at least $K_{\max}-1$ cells is maximal whenever that block's $K_{\max}$ is known
exactly. (d) A certified partition need not be maximal.
\end{proposition}

\begin{proof}
(a) Splitting one cell of an admissible partition in two gives a strict admissible refinement, and no strict refinement of a maximal partition is distinguishable.

(b) Let $\pi$ be certified and suppose $\sigma$ is a distinguishable strict refinement. Some cell $A$ of $\pi$ is
split by $\sigma$ into parts $A_1,\dots,A_r$ with $r\ge2$. Let $\pi'$ split $A$ into $A_1$ and
$A_2\cup\dots\cup A_r$. Then $\pi'$ is admissible, is a one-cell split of $\pi$, and is a coarsening of $\sigma$,
hence distinguishable under the assumption of part (b). This contradicts the certificate.

(c) If the factor has $K_{\max}$ cells, a strict refinement would have more than $K_{\max}$ cells and cannot be
distinguishable. If it has $K_{\max}-1$ cells, a distinguishable strict refinement has exactly $K_{\max}$ cells,
one more than the factor, and is therefore a one-cell split, which the certificate excludes. If every block factor of a certified partition meets this condition, the whole partition is maximal by
Proposition \ref{sprop:blocks}(iii).

(d) Take one block of four categories $a,b,c,d$ with estimated shares 4, 5, 6 and 7 percentage points, and let the
estimated covariance matrix of the four share estimators be $0.2^2\mathbf I+2.8^2\mathbf 1\mathbf 1^\top$. It arises from 12 strata of two PSUs each, where the two PSUs of a stratum hold equal shares of the frame. In strata 1--4 the first PSU holds 0.2 percentage points more of $a$, $b$, $c$ and $d$ respectively than the second, in stratum 5 it holds 2.8 points more of each of the four, and the two PSUs of every other stratum are identical. The common component $2.8^2\mathbf 1\mathbf 1^\top$ describes PSUs that differ in how much of the frame falls in the block as a whole. This component is made large, a correlation of 0.99 between the share estimators, since the example only has to show that part (d) can occur. With $c=3.0514$ and any cap $\bar m\ge2$, Table \ref{tab:toy}
shows that $\{a,b\}\,|\,\{c,d\}$ is distinguishable and certified, because both of its one-cell splits fail. Yet
the partition into singletons is distinguishable, so $\{a,b\}\,|\,\{c,d\}$ is not maximal. In this block the
singletons form the only maximal partition.
\end{proof}

\begin{table}[t]
\centering
\caption{The four-category example of Proposition \ref{sprop:cert}(d): separation ratios against $c=3.0514$.}
\label{tab:toy}
\begin{tabular}{llc}
\toprule
Partition & Pairs and separation ratios & Distinguishable \\
\midrule
$\{a,b\}\,|\,\{c,d\}$ & $(ab,cd)$ 10.000 & yes \\
$\{a\}\,|\,\{b\}\,|\,\{c,d\}$ & $(a,b)$ 3.536, $(a,cd)$ 3.190, $(b,cd)$ 2.836 & no \\
$\{a,b\}\,|\,\{c\}\,|\,\{d\}$ & $(ab,c)$ 1.063, $(ab,d)$ 0.709, $(c,d)$ 3.536 & no \\
$\{a\}\,|\,\{b\}\,|\,\{c\}\,|\,\{d\}$ & smallest of six ratios 3.536 & yes \\
\bottomrule
\end{tabular}
\end{table}

\begin{proposition}[Proposition~\ref{prop:nonmono} of the main text]\label{sprop:nonmono}
A distinguishable partition can have an admissible coarsening that is not distinguishable.
\end{proposition}

\begin{proof}
In Table \ref{tab:toy} the singletons are distinguishable while their coarsening $\{a\}\,|\,\{b\}\,|\,\{c,d\}$ is
not. Section 4 of the main text gives an instance in survey data.
\end{proof}

Merging can lower a separation ratio in two ways. For a cell $A$ compared with a cell $B$, merging a cell $C$ into $B$ shifts the gap by the share of $C$, since $\hat\Delta(A,B\cup C)=\hat\Delta(A,B)-\hat p_C$. The variance changes to $\hat v(A,B)+\hat v_C-2\,\widehat{\mathrm{cov}}(\hat p_A-\hat p_B,\hat p_C)$, where $\hat v_C$ is the variance of $\hat p_C$. The variance of the contrast with $B\cup C$ rises whenever the share of $C$ co-varies with $B$ at least as much as with $A$. In the example, the common component cancels in contrasts between cells with the same number of categories, such as $(ab,cd)$, and adds to the variance of contrasts between cells with different numbers of categories, such as $(ab,c)$ and $(b,cd)$, both of which fail.

\begin{proposition}[maximal partitions]\label{sprop:antichain}
(i) Two distinct maximal partitions never refine one another. (ii) In a block whose
$K_{\max}=K^\ast$ is known exactly, a distinguishable partition with $K^\ast-1$ cells is maximal unless it is a
one-step coarsening of a distinguishable $K^\ast$-cell partition. Each $K^\ast$-cell partition has
$\binom{K^\ast}{2}$ one-step coarsenings, so at least $N_{K^\ast-1}-\binom{K^\ast}{2}N_{K^\ast}$ of the $N_{K^\ast-1}$
distinguishable $(K^\ast-1)$-cell partitions are maximal, where $N_{K^\ast}$ is the number of distinguishable $K^\ast$-cell partitions.
\end{proposition}

\begin{proof}
(i) If maximal $\pi\ne\sigma$ and $\pi$ refines $\sigma$, then $\sigma$ has a distinguishable strict refinement, which contradicts the assumption that $\sigma$ is maximal.
(ii) A strict refinement of a partition with $K^\ast-1$ cells has at least $K^\ast$ cells. If it is distinguishable it
has exactly $K^\ast$, so it arises by splitting one cell in two, and the partition is one of its one-step
coarsenings. Counting those coarsenings gives the bound.
\end{proof}

Maximal partitions with different numbers of cells coexist whether or not coarsening can destroy distinguishability, so they are not caused by this loss of distinguishability. Because coarsening can destroy distinguishability, a certified partition need not be maximal, $K_{\max}$ cannot be inferred from the absence of distinguishable partitions with $K_{\max}+1$ cells alone, and a coarsening of a distinguishable partition has to be checked in its own right.

\begin{proposition}[singletons]\label{sprop:single}
If category $k$ is a singleton in a distinguishable partition, it is a singleton in every refinement of it, and
every distinguishable partition has at least one maximal refinement. Therefore, for any collection of
distinguishable partitions, the categories merged in every member include the categories merged in every maximal
partition.
\end{proposition}

\begin{proof}
A refinement only splits cells, and a singleton cannot be split. The distinguishable refinements of a partition
form a finite nonempty set, which has a maximal element under refinement. If $k$ is a singleton in some member,
it is a singleton in that member's maximal refinement.
\end{proof}

Unless the members of a collection are themselves maximal, the categories that are singletons in every member need not include the categories that are singletons in every maximal partition.

\begin{proposition}[Proposition~\ref{prop:enum} of the main text]\label{sprop:enum}
The walk of Section 3.2 visits every admissible partition of a block exactly once, and the walk that abandons a branch at its first indistinguishable pair returns exactly the block's distinguishable partitions under scope W.
\end{proposition}

\begin{proof}
At each step the walk holds a set $U$ of categories not yet in a cell and extends the partial partition by a new admissible cell
$S\subseteq U$ with $\min U\in S$. Every admissible partition of the block arises along exactly one path. At each step, the only choice consistent with a given partition is that partition's cell containing $\min U$. Where two paths first differ, they choose different cells for the same category, so they end in different partitions. Cells already chosen are never
altered. If two chosen cells fail, both remain cells of every completion, and by Lemma \ref{slem:local} the pair fails in each of them. Since scope W requires that pair to be distinguishable, abandoning the branch removes only partitions that are not distinguishable.
Conversely, each pair of cells of an output was tested when the later of the two was chosen, so every output is
distinguishable.
\end{proof}

The cost is governed by the $2^b$ cell totals, computed once, and by the number of pair tests the walk makes. In the illustration the walks made 2{,}245{,}892 and 81{,}303{,}674 pair tests for the vehicle and environment blocks, which have 27{,}644{,}437 and 10{,}480{,}142{,}147 partitions.

\begin{proposition}[Proposition~\ref{prop:orphan} of the main text]\label{sprop:orphan}
Let $F$ be an unmergeable category. Under scope A, if some other category $k$ lies in no admissible cell that
is distinguishable from $\{F\}$, no admissible partition is distinguishable.
\end{proposition}

\begin{proof}
Every admissible partition has $\{F\}$ as a cell and some cell $S\ni k$ with $S$ inside $k$'s block. Scope A
requires $S$ to be distinguishable from $\{F\}$, and no such $S$ exists.
\end{proof}

Every category other than $F$ must lie in some admissible cell distinguishable from $\{F\}$ for a partition to be distinguishable under scope A, but this condition is not sufficient. The scan can therefore only rule partitions out, which is all that Proposition~\ref{sprop:orphan} claims. For one category $k$ it tests the $2^{b-1}$ subsets of the block that contain $k$, and checking every category of the block at once tests all $2^b-1$ nonempty subsets.

\begin{proposition}[a lower bound on $K_{\max}$]\label{sprop:witness}
An explicit partition with $K$ cells, verified distinguishable pair by pair from the PSU totals, shows
$K_{\max}\ge K$ without any search.
\end{proposition}

\section{Algorithms}

\noindent\textbf{Algorithm 1 (search).} Input: categories, blocks, criterion cut $c$, scope.
\begin{enumerate}
\item Start from the all-singletons partition, with every unmergeable category a singleton.
\item Compute $r(A,B)$ from the PSU totals for every in-scope pair of cells that could be merged.
\item If no such pair fails the criterion, stop. Otherwise merge the failing pair with the smallest $r$, reading $r$ as
0 when $\widehat{\mathrm{SE}}=0$ and breaking ties in a fixed order. Then recompute every contrast that involves the new cell, and repeat this step.
\end{enumerate}

\noindent\textbf{Algorithm 2 (certificate, cap $\bar m$).}
\begin{enumerate}
\item For each cell $A$ of a distinguishable partition with $2\le|A|\le\bar m$, and each of its $2^{|A|-1}-1$ splits into two nonempty parts, test
whether the partition with $A$ replaced by the two parts is distinguishable.
\item If some split passes, apply the first one in a fixed order and restart step 1. The loop ends because each
applied split strictly refines the partition.
\item If a cell has more than $\bar m$ categories, report the partition as uncertified and do not count it as
maximal.
\end{enumerate}

\noindent\textbf{Algorithm 3 (listing every distinguishable partition of a block).}
\begin{enumerate}
\item Compute the PSU totals of all $2^b$ subsets of the block.
\item Walk$(\text{cells},U)$: if $U$ is empty, output the cells. Otherwise let $k=\min U$ and, for every admissible
$S\subseteq U$ with $k\in S$, test $S$ against each chosen cell. If all tests pass, call
Walk$(\text{cells}\cup\{S\},U\setminus S)$.
\item Call Walk$(\emptyset,\text{block})$. The outputs are the block's distinguishable partitions, from which $K_{\max}$,
the counts by number of cells and the maximal partitions follow.
\end{enumerate}

\noindent\textbf{Algorithm 4 (orphan scan for scope A).} For each block and each admissible cell $S\ne\{F\}$ in
it, test $S$ against the unmergeable cell $\{F\}$. Mark every category that lies in a passing subset. Any unmarked
category other than $F$ is an orphan, and one orphan shows that no admissible partition is distinguishable under scope A.

\section{The Illustration in Detail}

\paragraph{Design.} The survey's qualification rules and field period are NHTSA's own \citep{NHTSA2008a}, and the counts come from its analytical user's manual \citep{NHTSA2008b}. Of the 6{,}949 investigated crashes, 5{,}470 carry positive weight
and represent 2{,}188{,}970 crashes from July 3, 2005, to December 31, 2007. The 12 strata hold two PSUs each
\citep{Choi2008}, so $\nu=12$. The critical reason is recorded per vehicle. The code saying that a vehicle does not carry the crash's critical reason, and a placeholder code, are not crash-level categories. When two vehicles of a crash carry the critical reason, the lower-numbered vehicle's code is used. \citet[Sections 6, 9.2 and 10.1]{Choi2008} document the time blocks sampled and monitored (each a period of the day on one day of the week), how response varied across PSUs, and how the design weights are adjusted for crashes missed because a researcher missed the notification or arrived after the scene had been cleared.

\paragraph{Taxonomy.} The 59 categories observed at crash level form four blocks: driver codes
100--199 (29 observed), vehicle codes 200--299 (13), environment codes 500--527
(16), and the undetermined code 9999 \citep{NHTSA2008b}, which is never merged. Code numbers are labels, so any subset of a
block is an admissible cell.

\paragraph{Criteria.} The quantiles $t_{12,0.975}=2.1788$ and $t_{12,0.80}=0.8726$ give $c_1=3.0514$ at power 0.80, and with $t_{12,0.90}$ in place of $t_{12,0.80}$ the cut at power 0.90 is $c_1=3.5350$. D2 uses $c_2=2.1788$, and D3, with $q=0.05$, uses $c_3=2.436$ at the cell level.

\paragraph{The counterexample.} The vehicle partition shown in Section 4 of the main text is
\[
\{200, 202, 204, 205, 210, 211, 212, 213, 299\}\;|\;\{201, 206, 208\}\;|\;\{203\},
\]
with separation ratios 4.139, 5.615 and 4.288. Merging the last two cells gives
\[
\{200, 202, 204, 205, 210, 211, 212, 213, 299\}\;|\;\{201, 203, 206, 208\},
\]
whose one pair separates at 2.519.

\paragraph{The orphan scan.} Every criterion leaves at least one orphan in each of the vehicle and environment blocks, and thus no admissible partition is distinguishable under scope A for any criterion (Table \ref{tab:orphan} of the main text). The driver block, with $7.13\times 10^{22}$ partitions and $2^{29}$ subsets, is never
scanned, because one orphan in a smaller block is already enough.

\end{document}